\documentclass[11pt]{article}
\usepackage[lmargin=1in,rmargin=1in,tmargin=1in,bmargin=1in]{geometry}

\usepackage{times}
\usepackage{soul}
\usepackage{url}
\usepackage[hidelinks]{hyperref}
\usepackage[utf8]{inputenc}
\usepackage[small]{caption}
\usepackage{graphicx}
\usepackage{amsmath}
\usepackage{amsthm}
\usepackage{booktabs}
\usepackage[utf8]{inputenc}
\usepackage{pifont}
\usepackage{booktabs}
\usepackage{nicefrac}
\usepackage{mathrsfs,amsfonts,dsfont,authblk}

\usepackage[usenames,svgnames,table]{xcolor}
\usepackage{hyperref}[backref=page] 
\hypersetup{
colorlinks   = true,
linkcolor    = Red, 
urlcolor     = Blue, 
citecolor    = Blue 
}

\usepackage[capitalise,noabbrev]{cleveref}
\Crefname{remark}{Remark}{Remarks}
\Crefname{rmk}{Remark}{Remarks}
\Crefname{dfn}{Definition}{Definitions}
\Crefname{thm}{Theorem}{Theorems}
\Crefname{cor}{Corollary}{Corollaries}
\Crefname{lem}{Lemma}{Lemmas}
\Crefname{examplex}{Example}{Examples}
\Crefname{prop}{Proposition}{Propositions}

\usepackage{natbib}

\usepackage{bbm}

\usepackage{enumitem}
\usepackage{dsfont}
\usepackage{mathtools}
\usepackage{amsmath}
\usepackage{amsthm}
\usepackage{amssymb} 
\usepackage{afterpage}
\usepackage{array}
\usepackage{bbm}
\usepackage[normalem]{ulem} 
\usepackage{xcolor} 
\usepackage{subcaption}
\usepackage{accents} 

\usepackage[switch]{lineno} 

\newcount\Comments  
\newcommand{\kibitz}[2]{\ifnum\Comments=1{\color{#1}{#2}}\fi}

\newtheorem{thm}{Theorem}
\newtheorem{dfn}{Definition}
\newtheorem{lem}{Lemma}
\newtheorem{ex}{Example}
\newtheorem{prop}{Proposition}
\newtheorem{coro}{Corollary}
\newcommand{\mo}{\mathcal O}
\newcommand{\calP}{\mathcal P}
\newcommand{\calD}{\mathcal D}
\newcommand{\calL}{\mathcal L}
\newcommand{\calC}{\mathcal C}

\newcommand{\calW}{\mathcal W}

\DeclareMathOperator*{\argmax}{arg\,max}

\title{Convergence of Multi-Issue Iterative Voting under Uncertainty}

\author[1]{Joshua Kavner}
\author[2]{Reshef Meir}
\author[3]{Francesca Rossi}
\author[1]{Lirong Xia}
\affil[2]{Rensselaer Polytechnic Institute\\
	{\small\texttt{\{kavnej@rpi.edu,xial@cs.rpi.edu\}}}}
\affil[3]{Technion Israel Institute of Technology\\
	{\small\texttt{reshefm@ie.technion.ac.il}}}
\affil[4]{IBM T.J. Watson Research Center\\
	{\small\texttt{Francesca.Rossi2@ibm.com}}}

\begin{document}

\maketitle

\begin{abstract}
We study the effect of strategic behavior in iterative voting for multiple issues under uncertainty. We introduce a model synthesizing simultaneous multi-issue voting with \citet{Meir14:Local}'s local dominance theory and determine its convergence properties. 
After demonstrating that local dominance improvement dynamics may fail to converge, we present two sufficient model refinements that guarantee convergence from any initial vote profile for binary issues: constraining agents to have $\mathcal{O}$-legal preferences and endowing agents with less uncertainty about issues they are modifying than others. Our empirical studies demonstrate that although cycles are common when agents have no uncertainty, introducing uncertainty makes convergence almost guaranteed in practice.
\end{abstract}

\section{Introduction}

Consider a wedding planner who is deciding a wedding's banquet and wants to accommodate the party invitees' preferences. Suppose they have three issues: the main course (chicken, beef, or salmon), the paired wine (red or white), and the dessert cake's flavor (chocolate or vanilla). How should the planner proceed? 
On the one hand, the planner could request each attendee's (agent's) full preference ranking over the $d^p$ possibilities (alternatives), for $p$ issues with $d$ candidates each. However, this is computationally prohibitive and imposes a high cognitive cost for agents. 

On the other hand, the planner could only request agents' most preferred alternatives and decide each issue simultaneously. Although simpler, this option admits \emph{multiple election paradoxes} whereby agents can collectively select each of their least favored outcomes. For example, suppose three agents prefer $(1,1,0)$, $(1,0,1)$, and $(0,1,1)$ first respectively on three binary issues and all prefer $(1,1,1)$ last \cite{LacyNiou00}. Then the agents select $(1,1,1$) by majority rule on each issue independently.

A third approach is to decide the issues in sequence and have agents vote for their preferred alternative on the \emph{current} issue given the previously chosen outcomes \cite{Lang09:Sequential,Xia11:Strategic}. Still, the joint outcome may depend on the voting agenda and agents may be uneasy voting on the current issue if their preference depends on the outcomes of later issues \cite{Conitzer09:How}.

In this work, we study \emph{iterative voting} (IV) as a different yet natural method for deciding multiple issues \cite{Meir10:Convergence}. We elicit agents' most preferred alternatives and, given information about others' votes, allow agents to update their reports before finalizing the group decision. This approach combines the efficiency of simultaneous voting with the dynamics of sequential voting, thus incorporating information about agents' lower-ranked preferences without directly eliciting them. Like the former approach, agents only report their most preferred alternative. Like the latter approach, agents only update one issue at a time but are unrestricted in the order of improvements.

IV is an effective framework for its adaptability to various information and behavioral schemes. For example, agents may be fully aware of each other's votes, such as in online Doodle polls \cite{Zou15:Doodle}, and update to the myopic \emph{best response} of all others. Alternatively, agents may only have partial information, such as from imprecise opinion polls \cite{Reijngoud2012VoterRT} or latency if they can only periodically retrieve vote counts on a networked system.
We then take a strict uncertainty approach in which agents assume that any vote profile in a set is possible without assigning probabilities to them. Agents update to votes that \emph{locally dominate} their prior reports \cite{Meir14:Local}.

We ask two primary questions:
(1) \emph{Under what conditions does multi-issue IV converge?}
(2) \emph{How does introducing and increasing uncertainty affect the rate of convergence?}
The conclusion from previous research regarding convergence in single-issue IV is mixed, as the plurality and veto voting rules have strong guarantees, yet many other rules admit cycles \cite{meir2017iterative}.
This leaves us with mixed hope that multi-issue IV can converge, and if so, that it can solve other problems like multiple election paradoxes \cite{Xia11:Strategic}.
Furthermore, in contrast to the single-issue setting \cite{Meir15:Plurality}, uncertainty for multiple issues plays a double role. First, like the single-issue case, agents consider themselves as possibly pivotal on any issue that is sufficiently close to a tie. Second---and this part is new---agents may be uncertain regarding their own preferences on a particular issue, as this may depend on the outcomes of other uncertain issues.

\subsection{Our Contribution}
\label{sec:our_contribution}

On the conceptual side, we introduce a model that synthesizes prior work in local dominance strategic behavior, iterative voting, and simultaneous voting over multiple issues. This generalized model naturally captures both types of uncertainty discussed above. 



On the technical side, we first show that IV for either best response, without uncertainty, or local dominance improvement dynamics, with uncertainty, may not converge.
We then present two model refinements that prove sufficient to guarantee convergence for binary issues:
restricting agent preferences to be $\mo$-legal converges  because agents' preferences on issues are not interdependent; alternating uncertainty,
in which agents are more certain about the current issue than others, converges because fewer preference rankings yield improvement steps in this case.
These convergence results do not extend to the multi-candidate issues setting, as LDI plurality dynamics may cycle if agents have partial order preference information.




Our convergence results for binary issues also hold for a nonatomic variant of IV, in which agents are part of a very large population and arbitrary subsets of agents may change their vote simultaneously, establishing more general convergence results (see Appendix \ref{apx:non_atomic_model}).

We conclude with empirical evidence corroborating our findings that introducing uncertainty eliminates almost all cycles in IV for multiple binary issues. 
Our experiments further suggest IV improves the quality of equilibrium vote profiles relative to their respective truthful profiles, thus reducing multiple election paradoxes.
Increasing uncertainty yields faster convergence but degrades this welfare improvement. 

\subsection{Related Work}

Our model pulls insights from research across multi-issue voting, IV, and local dominance strategic behavior. Multi-issue voting is an extensively studied problem in economics and computer science with applications in direct democracy referendums, group planning and committee elections (see e.g., \citet{Lang16:Voting} for a survey). Our work follows research in agent strategic behavior by \citet{Lang07:Vote}, \citet{Lang09:Sequential}, \citet{Conitzer09:How}, \citet{Xia11:Strategic}.

Single-issue IV  was initially studied by \citet{Meir10:Convergence} for best response dynamics and the plurality social choice rule, whose authors bounded its guaranteed convergence rate. Subsequent work demonstrated that iterative veto converges \cite{Reyhani12:Best,Lev12:Convergence}, although many other voting rules do not \cite{Koolyk17:Convergence} unless agents are additionally restricted in their behavior \cite{Reijngoud2012VoterRT,Grandi13:Restricted,Obraztsova15:Convergence,Endriss16:Strategic}. 

This review already narrows down the possibility of convergence in the multi-issue setting. We therefore restrict our attention to plurality, extending the models of \citet{Meir14:Local}, \citet{Meir15:Plurality}.
Their research finds broad conditions for voting equilibrium to exist and guaranteed convergence for iterative plurality. In particular, \citet{Meir15:Plurality} studies a nonatomic model variation where agents have negligible impact on the outcome but multiple agents update their votes simultaneously, greatly simplifying the dynamics.



\citet{bowman14:potential} and \citet{grandi20:voting} empirically show for multiple binary issues that IV improves the social welfare of voting outcomes using
computer simulations and human experiments respectively. Our work augments this research by characterizing convergence in settings where agents do not have complete information. 


A related line of research studies  strategic behavior in epistemic voting games, when agents have uncertainty about the preferences or votes of others (see e.g., \citet{meir2018strategic} for a survey). Notably, \citet{Chopra04:Knowledge}, \citet{Conitzer11:Dominating}, \citet{Reijngoud2012VoterRT}, \citet{Ditmarsch13:strategic} each study the susceptibility and computational complexity of voting rules to local dominance improvement steps. 
Game-theoretic properties of strategic behavior for Gibbard-Satterthwaite games were studied by \citet{myerson1993theory}, \citet{grandi19:gibbard}, \citet{elkind20:cognitive}. Other forms of IV include work from \citet{airiau09:iterated}, \citet{Desmedt10:Equilibria}, \citet{Xia10:Stackelberg}.

\section{Preliminaries}
\label{sec:prelims}

\begin{paragraph}{Basic model.}
Let $\calP = \{1,2,\ldots,p\}$ be the set of $p$ issues over the joint domain $\calD = \times_{i=1}^p D_i$, where $D_i$ is the finite value domain of \emph{candidates} for issue $i$. We call the issues \emph{binary} if $D_i=\{0,1\}$ for each $i \in \calP$ or \emph{multi-candidate} otherwise.
Each of $n$ agents is endowed with a preference \emph{ranking} $R_j \in \calL(\calD)$, the set of strict linear orders over the $\times_{i=1}^p |D_i|$ alternatives. 
We call the collection of agents' preferences $P = (R_1, \ldots, R_n)$ a \emph{preference profile} and each agent's most preferred alternative their \emph{truthful} vote.
A \emph{vote profile} $a = (a_1, \ldots, a_n) \in \calD^n$ is a collection of \emph{votes}, where $a_j = (a^1_j, \ldots, a^p_j) \in \calD$ collects agent $j$'s  single-candidate vote per issue.
A \emph{resolute voting rule} $f:\calD^n \rightarrow \calD$ maps vote profiles to a unique outcome. 
We call $a \in \calD$ and $a^i \in D_i$ for $i \in \calP$ an \emph{alternative} or \emph{outcome} synonymously.
\end{paragraph}

\begin{paragraph}{Simultaneous plurality voting.}
A local voting rule, applied to each issue independently, is \emph{simultaneous} if issues' outcomes are revealed to agents at the same time. It is \emph{sequential} according to the order $\mo = \{o_1, \ldots, o_p\}$ if outcomes of each issue $o_i$ are revealed to agents prior to voting on the next issue $o_{i+1}$ \cite{LacyNiou00}. We focus on simultaneous plurality voting and adapt the framework of \citet{Xia11:Strategic}.

The plurality rule $f^i(a)$ applied to vote profile $a$ on issue $i$ only depends on the total number of votes for each of its candidates. We thus take the score of a candidate $c \in D_i$ as $s^i(c;a) = |\{j \leq n~:~a^i_j = c\}|$ and define the \emph{score tuple} $s(a) = (s^i(a))_{i \in \calP}$ as a collection of \emph{score vectors} $s^i(a) = (s^i(c;a))_{c \in D_i}$. We use the plurality rule $f(a) = (f^i(a))_{i \in \calP} \in \calD$, where $f^i(a) = \argmax_{c \in D_i} s^i(c;a)$, breaking ties lexicographically on each issue. We often use $s$, $s(a)$, and $a$ interchangeably for ease of notation.

Let $a_{-j}$ denote the vote profile without agent $j$ and $(a_{-j},\hat{a}_j)$ the profile $a$ by replacing $j$'s vote with the prospective vote $\hat{a}_j$. Then $s_{-j}$ and $s_{-j}+\hat{a}_j$ denote the corresponding adjusted score tuples.
\end{paragraph}

\begin{paragraph}{Preferential dependence.} 

A preference ranking is called \emph{separable} if its relative ordering of candidates in $D_i$, for each issue $i \in \calP$, are consistent across all outcomes of the other issues. 
This has the advantage that
such agents may express their preferences for any issue without knowing other issues' outcomes and avoid multiple-election paradoxes (see e.g., \citet{Xia11:Strategic}), but it is a very demanding assumption \cite{hodge2002separable}. Relaxing agents preferences to be $\mo$-legal maintains representation compactness without permitting arbitrary preferential dependencies \cite{Lang09:Sequential}.

Formally, for some order $\mo=\{o_1, \ldots, o_p\}$ over the issues, the preference ranking $R$ is called \emph{$\mo$-legal} if, given the outcomes of the prior issues $\{f^{o_1}, f^{o_2}, \ldots, f^{o_{i-1}}\}$, the relative ordering of candidates in $D_{o_i}$ is constant for any combination of outcomes of the later issues $\{f^{o_{i+1}},  \ldots, f^{o_p}\}$. Hence, the ranking of candidates $D_{o_i}$ in $R$ depends only, if at all, on the outcomes of issues prior to it in $\mo$.

The preference profile $P$ is called \emph{$\mo$-legal} if every ranking is $\mo$-legal for the same order $\mo$. 
Thus, the ranking $R$ is \emph{separable} if it is $\mo$-legal for \emph{any} order $\mo$. 

\end{paragraph}

\begin{ex}
Let there be $p=2$ binary issues and $n=3$ agents with preferences $P=(R_1, R_2, R_3)$ such that: 

$R_1 : (1,0) \succ_1 (0,0) \succ_1 (0,1) \succ_1 (1,1)$, 

$R_2 : (1,1) \succ_2 (0,0) \succ_2 (0,1) \succ_2 (1,0)$, and 

$R_3 : (0,0) \succ_3 (0,1) \succ_3 (1,0) \succ_3 (1,1)$.

The truthful vote profile $a = ( (1,0), (1,1),$ $(0,0) )$ consists of each agent's most preferred alternative.
The score tuple is then $s(a) = \{(1,2),(2,1)\}$, so with plurality $f(a)=(1,0)$.

Note that $R_1$ is $\mo$-legal for $\mo = \{2,1\}$: the agent always prefers $0 \succ 1$ on the second issue, yet their preference for the first issue depends on $f^2$. $R_3$ is separable, as the agent prefers $0 \succ 1$ on each issue independent of the other issue's outcome. $R_2$ is neither separable nor $\mo$-legal for any $\mo$.

Finally, agent $2$ can improve the outcome for themselves by voting for $\hat{a}_2 = (0,1)$ instead of $a_2 = (1,1)$. The adjusted score tuple is $s_{-2} = \{(1,1),(2,0)\}$, so $s_{-2} + \hat{a}_2 = \{(2,1),(2,1)\}$ and $f(s_{-2} + \hat{a}_2) = (0,0) \succ_2 (1,0) = f(a)$.
\label{ex:prelims}
\end{ex}

\begin{paragraph}{Improvement dynamics.}
We implement the iterative voting (IV) model introduced by \citet{Meir10:Convergence} and refined for uncertainty by \citet{Meir14:Local}, \citet{Meir15:Plurality}.
Given agents' truthful preferences $P$ and an initial vote profile $a(0)$, we consider an iterative process of vote profiles $a(t) = (a_1(t), \ldots, a_n(t))$ that describe agents' reported votes over time $t \geq 0$. 
For each round $t$, a \emph{scheduler} $\phi$ chooses an agent $j$ to make an \emph{improvement step} over their prior vote $a_j(t)$ by applying a specified \emph{response function} $g_j:\calD^n \rightarrow \calD$. Each agent's response implicitly depends on their preferences and belief about the current \emph{real} vote profile, but they are not aware of others' private preferences. All other votes besides $j$'s remain unchanged.

A scheduler is broadly defined as a mapping from sequences of vote profiles to an agent with an improvement step in the latest vote profile \cite{Apt12:Classification}. An improvement step must be selected if one exists, and a vote profile where no improvement step exists (i.e., $g_j(a) = a_j ~\forall j \leq n$) is called an \emph{equilibrium}.
The literature on game dynamics considers different types of response functions, schedulers, initial profiles, and other assumptions (see e.g., \citet{fudenberg09:learning}, \citet{Marden07:Regret}, \citet{Bowling05:Convergence}, \citet{Young93}). This means that there are multiple levels in which a voting rule may guarantee convergence to an equilibrium \cite{meir2017iterative}. 
In this work, we study two response functions: \emph{best response (BR)}, without uncertainty, and \emph{local dominance improvements (LDI)}, with uncertainty. For both dynamics, we restrict agents to only changing
their vote on a single \emph{current} issue $i \in \calP$ per round, as determined by the scheduler $\phi$. 
We therefore have the following form of convergence, as described by \citet{kukushkin11:acyclicity}, \citet{Monderer96:Potential}, \citet{milchtaich96:congestion}:
\begin{dfn}
An IV dynamic has the \emph{restricted-finite improvement property} if every improvement sequence is finite from any initial vote profile for a given response function.
\label{dfn:convergence}
\end{dfn}

Under BR dynamics, agents know the real score tuple $s(a)$.

\begin{dfn}[Best response]
Given the vote profile $a$, $g_j(a) := \hat{a}_j$ which yields agent $j$'s most preferred outcome of the set $\{f(a_{-j}, \tilde{a}_j): \tilde{a}^i_j \in D_i, \tilde{a}^k_j = a^k_j~\forall k \neq i\}$, defaulting as $a_j$ if the best outcome is the same.
\label{dfn:best_response}
\end{dfn}

LDI dynamics are based on the notions of \emph{strict uncertainty} and \emph{local dominance} \cite{Conitzer11:Dominating,Reijngoud2012VoterRT}. Let $S \subseteq \times_{i=1}^p \mathbb{N}^{|D_i|}$ be a set of score tuples that, informally, describes agent $j$'s uncertainty about the real score tuple $s(a)$.
An LDI step to a prospective vote $\hat{a}_j$ is 
(1) weakly better off than their original $a_j$ for every $v \in S$, (2) strictly better off for some tuple in $v \in S$, and (3) not worse off for any  $v \in S$. This is formally defined as follows.

\begin{dfn}
The vote $\hat{a}_j$ \emph{$S$-beats} $a_j$ if there is at least one score tuple $v \in S$ such that $f(v+\hat{a}_j) \succ_j f(v+a_j)$.
The vote $\hat{a}_j$ \emph{$S$-dominates} $a_j$ if both (I) $\hat{a}_j$ $S$-beats $a_j$; and (II) $a_j$ does not $S$-beat $\hat{a}_j$.
\label{dfn:s_dom}
\end{dfn}

\begin{dfn}[Local dominance improvement]
Given the vote profile $a$ and agent $j$, let $LD^i_j$ be the set of votes that $S$-dominate $a_j$, only differ from $a_j$ on the $i^{th}$ issue, and are not themselves $S$-dominated by votes differing from $a_j$ only on the $i^{th}$ issue. Then
$g_j(a) = a_j$ if $LD^i_j = \emptyset$ and $\hat{a}_j \in LD^i_j$ otherwise.
\label{dfn:local_dom}
\end{dfn}

This definition distinguishes from (weak) LDI in \citet{Meir15:Plurality}, in that agents may change their votes consecutively but only on different issues.
Note that $S$-dominance is transitive and antisymmetric, but not complete, so an agent $j$ may not have an improvement step.
To fully define the model, we need to specify $S$ for every $a$. For example, if $S=\{s(a_{-j})\}$ and each $j$ has no uncertainty about the real score tuple, then LDI coincides with BR and an equilibrium coincides with \emph{Nash equilibrium}. Therefore, LDI broadens BR dynamics.
\end{paragraph}

\begin{paragraph}{Distance-based uncertainty.}
Agents in the single-issue model constructed their uncertainty sets using \emph{distance-based uncertainty}, in which all score vectors close enough to the current profile were believed possible \cite{Meir14:Local,Meir15:Plurality}. We adapt this to the multi-issue setting by assuming agents uphold candidate-wise distance-based uncertainty over score vectors for each issue independently.

For any issue $i \in \calP$, let $\delta(s^i(a),\tilde{s}^i(a))$ be a distance measure for score vectors for any vote profile $a$. This measure is \emph{candidate-wise} if it can be written as $\delta(s^i(a),\tilde{s}^i(a)) = \max_{c \in D_i} \hat{\delta}(s^i(c;a),\tilde{s}^i(c;a))$ for some monotone function $\hat{\delta}$. For example, the multiplicative distance (where $\hat{\delta}(s, \tilde{s}) = \max\{ \tilde{s}/s, s/\tilde{s}\}-1$) and the $\ell_\infty$ metric (where $\hat{\delta}(s, \tilde{s}) = |s-\tilde{s}|$) are candidate-wise.

Given the vote profile $a$ and issue $i \in \calP$, 
we model agent $j$'s uncertainty about the adjusted score vector $s^i_{-j}$ by the \emph{uncertainty score set} 
$\tilde{S}^i_{-j}(s;r^i_j) := \left\{v^i: \delta(v^i, s^i_{-j}) \leq r^i_j \right\}$ with an \emph{uncertainty parameter} $r^i_j$. That is, given all other agents' votes $a_{-j}^i$, agent $j$ is not sure what the real score vector is $v^i \in \tilde{S}_{-j}^i(s;r^i_j)$. We take $\tilde{S}_{-j}(s,r_j) := \times_{i=1}^p \tilde{S}^i_{-j}(s;r^i_j)$ for $r_j=(r_j^i)_{i \in \calP}$, and drop the parameters if the context is clear.


\end{paragraph}

\begin{ex}
Let there be $p=2$ binary issues and $n=13$ agents. Define the vote profile $a$  such that $s(a)=\{(8,5),(9,4)\}$ and $a_j = (0,1)$ for some agent $j$. Suppose $j$ uses the $\ell_\infty$ uncertainty metric with $(r^1_j, r^2_j) = (1,1)$. Then the uncertainty score set without $j$'s vote is:

$
\tilde{S}_{-j} = \{(6,7,8) \times (4,5,6)\} \times \{(8,9,10) \times (2,3,4)\}
$

Consider the prospective vote $\hat{a}_j = (1,1)$. Then the set of possible outcomes is
$
\{f(v+\hat{a}_j)~:~v \in \tilde{S}_{-j}\} = \{0,1\} \times \{0\}.
$



\label{ex:comb_LD}
\end{ex}

\begin{paragraph}{Remark.} Notice that in any vote profile reachable via BR or LDI dynamics from the truthful vote profile, no agent with fully separable preferences will have an improvements step. This follows from the definitions of BR and LDI.
\end{paragraph}

\section{Best-Response Dynamics}

Given the vote profile $a$, consider agent $j$ changing their vote $a_j$ on issue $i$ to the prospective vote $\hat{a}_j$. Under BR dynamics, without uncertainty, $j$ changes their vote only if they can feasibly improve the outcome $f(a)$ to one more favorable with respect to  $R_j$. This happens under two conditions. First, $j$ must be \emph{pivotal} on the $i^{th}$ issue, meaning that changing their vote will necessarily change the outcome.
Second, $j$ must be preferential to change $i$ by voting for $\hat{a}^i_j$ over $a^i_j$ given the outcomes of the other issues $\calP \backslash \{i\}$. Agent $j$'s preferences are always well-defined since they know every issue's real outcome. Thus BR dynamics behave similar to the single-issue setting, which we recall converges \cite{Meir10:Convergence}. However, in the multi-issue setting, agents' preferences on each issue may change as other issues' outcomes change. This entails the possibility of a cycle, as declared in the following proposition and proved with the subsequent example. 

\begin{prop}
BR dynamics for multiple issues may not converge, even if issues are binary.
\end{prop}

\begin{ex}
Let there be $p=2$ binary issues and $n=3$ agents without uncertainty and the following preferences:

$R_1: (0,1) \succ_1 (1,1) \succ_1 (1,0) \succ_1 (0,0)$, 

$R_2: (0,0) \succ_2 (0,1) \succ_2 (1,1) \succ_2 (1,0)$, and 

$R_3: (1,0) \succ_3 (1,1) \succ_3 (0,0) \succ_3 (0,1)$.


\begin{table}[h!]
\centering
 \begin{tabular}{|c|c|c|c|c|} 
 \hline
 Agent $j$ & $a_j(0)$ & $a_j(1)$ & $a_j(2)$ & $a_j(3)$ \\ [0.5ex]
 \hline
 $1$ & $(0,1)$ & $(1,1)$ & $(1,1)$ & $(0,1)$ \\
 $2$ & $(0,0)$ & $(0,0)$ & $(0,1)$ & $(0,1)$ \\
 $3$ & $(1,0)$ & $(1,0)$ & $(1,0)$ & $(1,0)$ \\
 \hline \hline
 $f(a)$ & $(0,0)$ & $(1,0)$ & $(1,1)$ & $(0,1)$ \\
 \hline
 \end{tabular}
 \caption{Agents' votes for $a(0)$ (truthful), $a(1)$, $a(2)$, and $a(3)$.}
 \label{tab:BR_cycle}
\end{table}

Table \ref{tab:BR_cycle} demonstrates a cycle generated by BR dynamics from the truthful vote profile $a(0)$. The order of improvement steps is $j=(1,2,1,2)$. No other BR step is possible from any profile in the cycle, demonstrating that no agent scheduler can lead to convergence. Notice that agent $3$ never changes their vote since their ranking is fully separable.


\label{ex:BR_cycle}
\end{ex}


\section{Local Dominance Improvement Dynamics}

\subsection{Non-Convergence}

LDI dynamics broadens best-response, but it is initially unclear how increasing agents' uncertainty parameters will affect their improvement steps and whether this eliminates the possibility of cycles. 
Under LDI dynamics, with uncertainty, agent $j$ changes their vote $a_j$ on issue $i$ to the prospective vote $\hat{a}_j$ from vote profile $a$ with two conditions similar to BR dynamics: if they (I) believe they may be pivotal and (II) can improve the outcome with respect to $R_j$ (corresponding to $S$-dominance). 
First, intuitively, it would seem that uncertainty over $j$’s current issue $i$ affects whether they believe they can change $i$'s outcome. Notice that if issue $i$ has a unique winning candidate and no runner-up, then $j$ has no best-response but may have an LDI step only if their uncertainty parameter is large enough. 
Second, unlike BR dynamics and the single-issue setting \cite{Meir14:Local,Meir15:Plurality}, $j$'s preference over candidates of issue $i$ may depend on the outcomes of other issues, which $j$ may be uncertain about. It stands to reason that the more uncertainty $j$ has over other issues, the less clarity the agent has over their own preference for issue $i$'s candidates. 

Despite these two apparent parallels, we observe for multi-candidate issues that the relationship between uncertainty parameter and existence of LDI steps is not monotonic. Example \ref{ex:LDI_steps_different_uncertainties} in Appendix \ref{apx:ex} demonstrates that different sets of prospective votes $LD^i_j$ in the same vote profile $a$ are non-comparable for the same agent with different uncertainty parameters. On the other hand, LDI dynamics have the aforementioned structure with binary issues, as we prove in Theorem \ref{thm:improvement_step_relations}. Still, this relationship has counter-acting forces: increasing uncertainty on issue $i$ may only add LDI steps over issue $i$ but may only eliminate steps on all other issues.



\begin{thm}
Given binary issues, consider agent $j$ changing their vote on issue $i$ in vote profile $a$ with one of three uncertainty parameters: $r_j$, $q_j$, or $\hat{q}_j$. The parameters $r_j$ and $q_j$ only differ on issue $k \neq i$ such that $r_j^k < q_j^k$; $r_j$ and $\hat{q}_j$ only differ on issue $i$ such that $r_j^i < \hat{q}_j^i$. Let $LD^i_j$, $LQ^i_j$, and $\hat{LQ}^i_j$ denote $j$'s possible LDI steps according to each of the parameters. Then $LQ_j^i \subseteq LD_j^i \subseteq \hat{LQ}^i_j$.

\label{thm:improvement_step_relations}
\end{thm}

We prove the theorem in two parts by demonstrating that if a vote $\hat{a}_j~\tilde{S}_{-j}(a;q_j)$-dominates $a_j$, then it must $\hat{a}_j~\tilde{S}_{-j}(a;r_j)$-dominates $a_j$; likewise, this implies that $\hat{a}_j~\tilde{S}_{-j}(a;\hat{q}_j)$-dominates $a_j$. Each of these relationships arise as a result of $\tilde{S}_{-j}(a;r_j) \subseteq \tilde{S}_{-j}(a;q_j)$ and $\tilde{S}_{-j}(a;r_j) \subseteq \tilde{S}_{-j}(a;\hat{q}_j)$. For binary issues, this is sufficient and we do not need to show that $\hat{a}_j$ is not $\tilde{S}_{-j}(a;r_j)$- or $\tilde{S}_{-j}(a;\hat{q}_j)$-dominated. This may not be the case for multi-candidate issues, as previously stated in Example \ref{ex:LDI_steps_different_uncertainties}. The full proof may be found in Appendix \ref{apx:convergence_proofs}. 


In sum, we find that enabling agents to have uncertainty over issues does not eliminate the possibility of a cycle, unlike the single-issue setting, as declared in the following proposition and proved with Example \ref{ex:LDI_cycle} in Appendix \ref{apx:ex}.

\begin{prop}
LDI dynamics with multiple issues may not converge, even if agents have the same constant uncertainty parameters and issues are binary.
\end{prop}


\subsection{Strategic responses and $\mo$-legal preferences}

We next present two model refinements that prove sufficient to guarantee convergence for binary issues: $\mo$-legal preferences and a specific form of dynamic uncertainty. 

We are motivated by observing in Examples \ref{ex:BR_cycle} and \ref{ex:LDI_cycle} that cycles appear due to agents' interdependent preferences among the issues. For simplicity, consider two issues $i$ and $k$ in which no agent has an LDI step on issue $i$. LDI steps on issue $k$ may change the set of possible winning candidates (see Definition \ref{dfn:possible_potential_winners}), thus transforming agents' preferences over the candidates in $i$. This enables agents to have LDI steps over issue $i$, even if they may no longer have improvement steps over $k$.
It stands to reason that eliminating interdependent preferences by restricting agents to have $\mo$-legal preference rankings would guarantee convergence. 

We prove this is the case in Theorem \ref{thm:o_legal_converge}. We first introduce a characterization about agents' strategic responses, extending a lemma from \citet{Meir15:Plurality} to the multi-issue setting.

\begin{dfn}
\label{dfn:possible_potential_winners}
Agent $j$ believes a candidate $c$ on issue $i$ is a \emph{possible winner} if there is some score vector where $c$ wins. Formally, 

$W^i_j(a) = \{c \in D_i~:~ \exists v \in \tilde{S}_{-j}(a; r_j)~\text{s.t.}~f^i(v+a_j)=c\}$.

In contrast, $j$ calls $c$ a \emph{potential winner} if there is some score vector in which they can vote to make $c$ win: 

$H^i_j(a) = \{c \in D_i~:~ \exists v \in \tilde{S}_{-j}(a; r_j)~\text{and}~\hat{a}_j~\text{s.t.}~\hat{a}^i_j = c, \hat{a}^k_j = a^k_j~\forall k \neq i,~\text{s.t.}~f^i(v+\hat{a}_j)=c\}$.

The set of real potential winners is denoted: 
$H^i_0(a) = \{c \in D_i~:~ f^i(s_{-j}+\hat{a}_j) = c ~ \text{where}~ \hat{a}^i_j=c, \hat{a}^k_j = a^k_j~\forall k \neq i\}$.
\end{dfn}


By this definition, $W^i_j(a) \subseteq H^i_j(a)$. \footnote{
Without uncertainty, $H^i_j(a)$ (or $H^i_j(a) \cup \{a^i_j\}$ if adding a vote to $a^i_j$ makes it win) is also known as the \emph{chasing set} (excluding $f(a)$) \cite{Rabinovich15:Analysis} or \emph{potential winner set} (including $f(a)$) \cite{Kavner21:strategic} on issue $i$. 
$H^i_j(a)$ coincides with \citet{Meir14:Local} and \citet{Meir15:Plurality}'s definition of a possible winner ``$W_j(s)$.''
}  Denote by $\calW^{-i}(a; r_j) = \times_{k \in \calP \backslash \{i\}} W^k_j(a)$ the set of possible winning candidates on all issues besides $i$, from agent $j$'s perspective with uncertainty parameter $r_j$.

\begin{lem}
Consider an LDI step $a_j \xrightarrow{j} \hat{a}_j$ over issue $i$ from vote profile $a$ by agent $j$ with uncertainty parameter $r_j$. Then either (1) $a^i_j \notin H^i_j(a)$; or for every combination of possible winners in $\calW^{-i}(a; r_j)$, either (2) $a^i_j \prec_j b$ for all $b \in H^i_j(a)$; or (3) $r^i_j = 0$, $\{a^i_j, \hat{a}^i_j\} \subseteq H_0(a)$ and $\hat{a}^i_j \succ_j a^i_j$.
\label{lem:lem3_multi}
\end{lem}

The proof of this lemma directly follows that of Lemma 3 in \citet{Meir15:Plurality}; see Appendix \ref{apx:convergence_proofs}.


\begin{thm}
Suppose all agents have $\mo$-legal preferences for the common order $\mo$ over binary issues.
Then LDI dynamics converge.
\label{thm:o_legal_converge}
\end{thm}

\begin{proof}
Fix an initial vote profile $a(0)$.
Suppose for contradiction that there is a cycle among the vote profiles $\calC = \{a(t_1), \ldots, a(t_T)\}$, where $a(t_T +1) = a(t_1)$ and $a(t_1)$ is reachable from $a(0)$ via LDI dynamics.
Let $i$ be the highest order issue in $\mo$ for which any agent changes their vote in $\calC$.

Let $t^* \in [t_1, t_T)$ be the first round that some agent $j$ takes an LDI step on issue $i$, where $a_j \xrightarrow{j} \hat{a}_j$ from vote profile $a(t^*)$; let $t^{**} \in (t^*, t_T]$ be the last round that $j$ switches their vote on $i$ back to $a^i_j$. It must be the case that $a^i_j \in H^i_j(a(t^*))$, since issues are binary and otherwise, $|H^i_j(a(t^*))|=1$ and $j$ would not have an improvement step. Hence by Lemma \ref{lem:lem3_multi}, $\hat{a}^i_j \succ_j a^i_j$ for every combination of possible winners in $\calW^{-i}(a(t^*); r_j)$. Likewise, on round $t^{**}$, $a^i_j \succ_j \hat{a}^i_j$ for every combination of possible winners in $\calW^{-i}(a(t^{**}); r_j)$. Thus for some issue $k$ and outcomes $x,y \in \{0,1\}$, $x \neq y$, we have $W^k_j(a(t^*)) = \{x\}$ and $W^k_j(a(t^{**})) = \{y\}$.

Since $j$ has $\mo$-legal preferences, $k$ must be prior to issue $i$ in the order $\mo$. However, no agent changed their vote on issue $k$ between rounds $t^*$ and $t^{**}$ so it must be that $x \in W^k_j(a(t^{**}))$, even if $j$'s uncertainty parameters changed. This forms a contradiction, so no such cycle can exist.

\end{proof}

The intuition behind Theorem \ref{thm:o_legal_converge} is that as an LDI sequence develops, there is some ``foremost'' issue $i$ in which no LDI step takes place on any issue prior to $i$ in the order $\mo$. Agents' relative preferences for the candidates in $i$ are fixed because their preferences are $\mo$-legal: score vectors for issues prior to $i$ in $\mo$ do not change scores of issues afterward do not affect agents' preferences for $i$. Hence, agents' improvement steps over the issue $i$ converge, whereas any cycle must have a sub-sequence of vote profile whose votes for issue $i$ cycles.

Note that $\mo$-legality is not necessary for convergence, as BR dynamics induced from the truthful vote profile in Example \ref{ex:prelims} converge. Although $\mo$-legality is a strict assumption, loosening this even slightly may lead to cycles. Example \ref{ex:BR_cycle} demonstrates a cycle in which each agent has an $\mo$-legal ranking but orders differ between agents.

Furthermore, the theorem describes that LDI steps over the issue $i$ eventually terminate, thus enabling each subsequent issue in $\mo$ to converge. 
This seems to suggest that IV under $\mo$-legal preferences is the same as \emph{truthful sequential voting}, where agents vote for their preferred alternative on each issue $o_i$ given the known previous outcomes of $\{o_1, \ldots, o_{i-1}\}$ \cite{Lang09:Sequential}.
Although the procedures' outcomes could be the same, there are two notable differences. First, the initial vote profile could have an issue with a single possible winning candidate that differs from the truthful sequential outcome; then no agent has an improvement step on that issue. Second, based on the scheduler, there may be a round with multiple possible winning candidates but no subsequent improvement steps on an issue. This may terminate IV with outcomes different than the truthful sequential vote.

This convergence result does not extend to the multi-candidate case, as declared in the following proposition and proved with the subsequent example.

\begin{prop}
LDI dynamics may not converge for multiple issues, even if agents have the same constant uncertainty parameters and $\mo$-legal preferences for the common order $\mo$.
\end{prop}

\begin{ex}
Let there be $p=2$ issues and $n=15$ agents who each use the $\ell_\infty$ uncertainty metric with common fixed uncertainty parameters $(r^1_j, r^2_j) = (2,1)~\forall j \leq n$. Label the candidates $\{0,1\}$ and $\{a,b,c,d\}$ respectively. Agent $j$ has preferences: if $f^1 = 0$ then $b \succ_j c \succ_j a \succ_j d$ on the second issue; otherwise $c \succ_j b \succ_j a \succ_j d$. Agent $k$ always prefers $a \succ_k d \succ_k b \succ_k c$ on the second issue. These preferences are $\mo$-legal for $\mo=\{1,2\}$.

Define $a(0)$ so $s(a(0)) = \{(7,8), (3,5,5,2)\}$ and $a_j(0) = a_k(0) = (0,a)$. There are four LDI steps involved in this cycle: (i) $(0,a) \xrightarrow{j} (0,d)$, (ii) $(0,a) \xrightarrow{k} (0,d)$, (iii) $(0,d) \xrightarrow{j} (0,a)$, and (iv) $(0,d) \xrightarrow{k} (0,a)$.
We prove these steps are valid in Appendix \ref{apx:ex}. 

\label{ex:o_legal_multicand_cycle}
\end{ex}

Note that $H^2_j(a(0)) = \{a,b,c\}$ and $H^2_j(a(2)) = \{b,c,d\}$. In contrast to the single-issue setting (see Lemma 4 of \citet{Meir15:Plurality}), agent $j$ takes LDI steps to candidates not in the potential winning set. This results from $j$'s uncertainty over whether $b$ or $c$ is most-preferred, even as both are preferable to $a$ and $d$. Hence, we get the following corollary:

\begin{coro}
LDI dynamics may not converge for plurality over a single issue for agents with partial order preferences.

\label{cor:cycle_partial_order}
\end{coro}

\subsection{Alternating Uncertainty}

In Theorem \ref{thm:improvement_step_relations} we found that for binary issues, agents may have fewer LDI steps over an issue $i$ if that issue has less uncertainty and other issues have more. This suggests that LDI steps occur from a relative lack of information about the current issue's score vector than for other issues. If agents can gather more information about the current issue before changing their vote, thereby decreasing its uncertainty relative to other issues, then they may not have an LDI step. 

We therefore consider a 
specific form of dynamics over agents' uncertainty parameters where agents can gather this information and consider themselves pivotal only with respect to the lowered uncertainty. Agents are assumed to subsequently forget this relative information since it may be outdated by the time they change their vote again. 
We show in the following theorem that this eliminates cycles. 

\begin{dfn}(Alternating Uncertainty.)
Fix two parameters $r^c_j$, $r^o_j$ for each agent $j$ such that $r^c_j < r^o_j$. Define each agent $j$'s uncertainty parameters such that whenever they are scheduled
to change their vote on issue $i$, $j$'s uncertainty for $i$ is $r^c_j$ and for each other issue $k \neq i$ is $r^o_j$.
\label{dfn:alternating_uncertainty}
\end{dfn}

\begin{thm}
Given binary issues, LDI dynamics converges for agents with alternating uncertainty.
\label{thm:dynamic_uncertainty_binary_converge}
\end{thm}

\begin{proof}

Fix an initial vote profile $a(0)$ and uncertainty parameters $r^c_j, r^o_j$ for each agent $j \leq n$. Suppose for contradiction that there is a cycle among the vote profiles $\calC = \{a(t_1), \ldots, a(t_T)\}$, where $a(t_T+1)=a(t_1)$ and $a(t_1)$ is reachable from $a(0)$ via LDI dynamics. Without loss of generality, suppose all issues and agents are involved in the cycle. 

Consider the agent $j$ with the largest $r^o_j = \max_{u \leq n} r^o_u$. Let $t^* \in [t_1, t_T)$ be the first round that $j$ takes an LDI step on issue $i$, where $a_j \xrightarrow{j} \hat{a}_j$ from vote profile $a(t^*)$; let $t^{**} \in (t^*, t_T]$ be the last round that $j$ switches their vote on $i$ back to $a^i_j$. It must be the case that $a^i_j \in H^i_j(a(t^*))$, since issues are binary and otherwise, $|H^i_j(a(t^*))|=1$ and $j$ would not have an improvement step. Hence by Lemma \ref{lem:lem3_multi}, $\hat{a}^i_j \succ_j a^i_j$ for every combination of possible winners in $\calW^{-i}(a(t^*); r_j)$. Likewise, on round $t^{**}$, $a^i_j \succ_j \hat{a}^i_j$ for every combination of possible winners in $\calW^{-i}(a(t^{**}); r_j)$. Thus for some issue $k$ and outcomes $x,y \in \{0,1\}$, $x \neq y$, we have $W^k_j(a(t^*)) = \{x\}$ and $W^k_j(a(t^{**})) = \{y\}$.

Let $t' \in (t^*, t^{**})$ be the first round since $t^*$ that some agent $h$ changes their vote on issue $k$. Then $H^k_h(a(t')) = \{0,1\}$. Since $W^k_j(a(t')) = W^k_j(a(t^*)) = \{x\} \subsetneq \{0,1\}$ and distance functions are candidate-wise, $r^c_h \geq r^o_j$. This entails $r^o_h > r^o_j$ by definition of alternating uncertainty, which contradicts the assertion that $j$ is the agent $u$ with the largest $r^o_u$.

\end{proof}

This convergence result does not extend to the multi-candidate case, as 
Example
\ref{ex:o_legal_multicand_cycle} also covers this setting.





\section{Experiments}
\label{sec:exps}





\begin{figure*}[t] 
  \begin{minipage}[b]{0.49\linewidth}
    \centering
    \includegraphics[width=\linewidth]{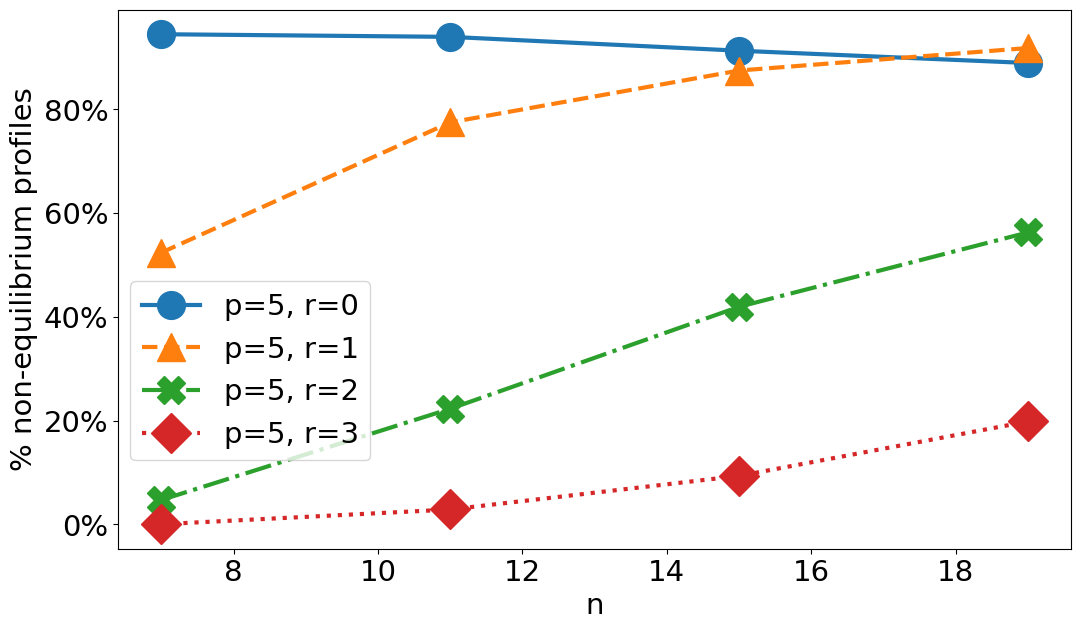}
    \caption{Percentage of non-equilibrium profiles as $n$\\increases.} 
    \label{fig:perc_dyns}
    \vspace{2ex}
  \end{minipage}
  \hspace{1ex}
  \begin{minipage}[b]{0.49\linewidth}
    \centering
    \includegraphics[width=\linewidth]{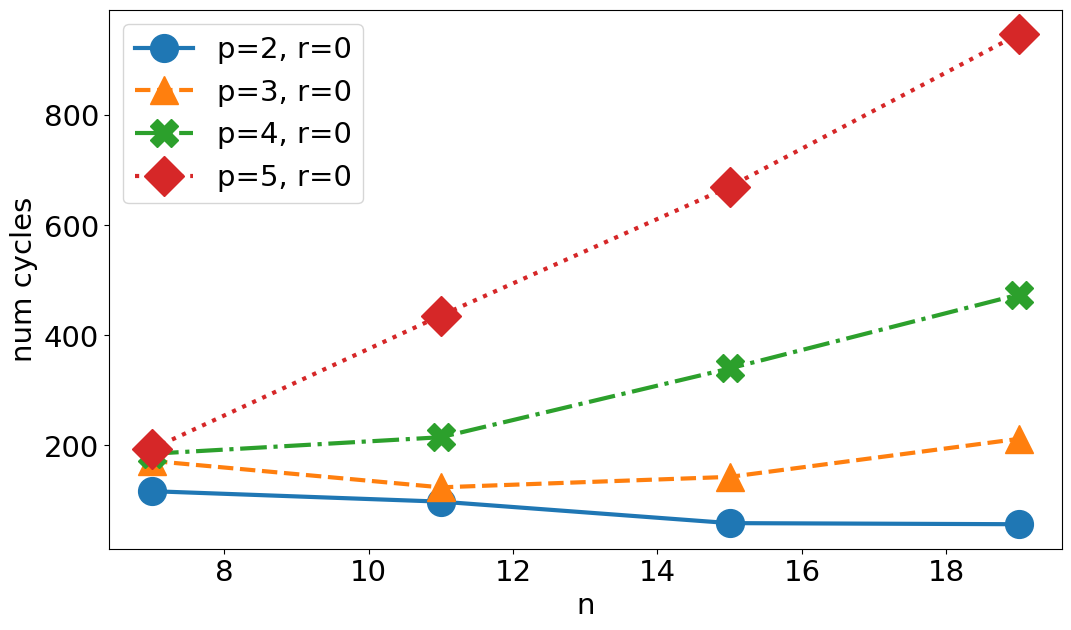}
    \caption{Number of profiles whose sequences cycle as $n$ increases.} 
    \label{fig:num_cycles}
    \vspace{2ex}
  \end{minipage} 
  \begin{minipage}[b]{0.49\linewidth}
    \centering
    \includegraphics[width=\linewidth]{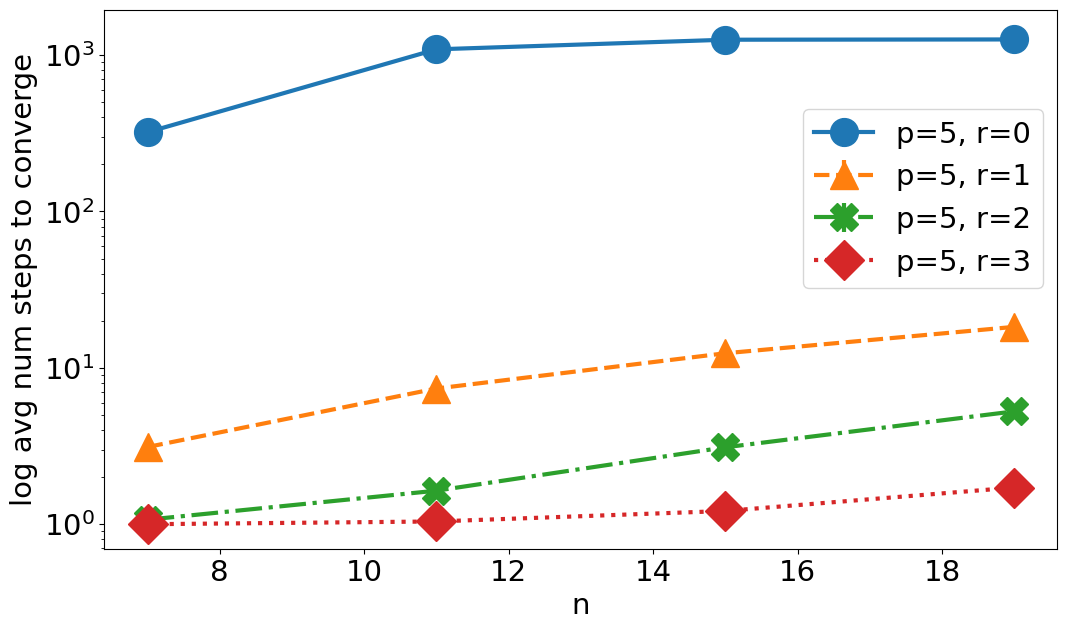}
    \caption{Log average number steps to converge as $n$\\increases; 95\% CI (too small to show).}
    \label{fig:steps_to_converge}
    \vspace{2ex}
  \end{minipage}
  \hspace{1ex}
  \begin{minipage}[b]{0.49\linewidth}
    \centering
    \includegraphics[width=\linewidth]{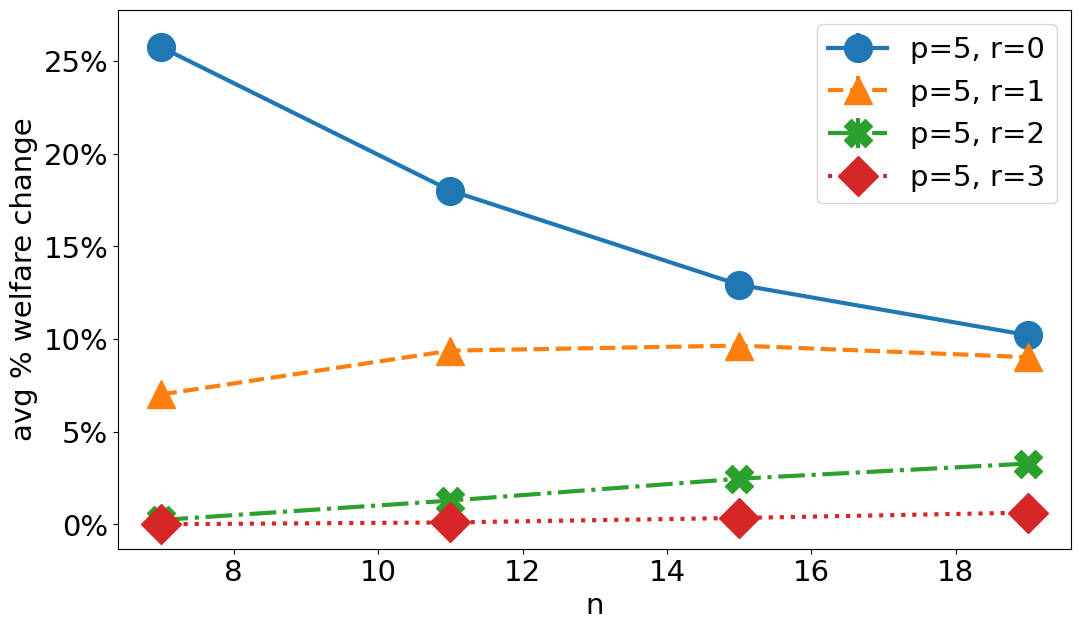}
    \caption{Average percent change in welfare as $n$\\increases; 95\% CI (too small to show).}
    \label{fig:welf_perc_change}
    \vspace{2ex}
  \end{minipage} 
\end{figure*}

Our computational experiments  investigate the effects of uncertainty, number of issues, and number of agents on LDI dynamics. Specifically, we ask how often truthful vote profiles are themselves in equilibrium, how often LDI dynamics does not converge, and the path length to equilibrium when it converges. 
Our inquiry focuses on whether cycles are commonplace in practice even as convergence is not guaranteed.



We answer these questions 
for a broad cross-section of inputs, with
$n \in \{7, 11, 15, 19\}$ agents, $p \in \{2,3,4,5\}$ binary issues, and $r \in \{0,1,2,3\}$ uncertainty that is constant for all agents, issues, and rounds. We generate $10,000$ preference profiles for each combination by sampling agents' preferences uniformly and independently at random. We simulate LDI dynamics from the truthful vote profile using a scheduler 
that selects profiles
uniformly at random from the set of valid LDI steps from all agents and issues. If there are no such steps, we say the sequence has converged. Otherwise, we take $50,000$ rounds as a sufficiently large stopping condition to declare the sequence has cycled.

Our results are presented in Figures \ref{fig:perc_dyns} -- \ref{fig:steps_to_converge} with respect to $n$ and $r$ for $p=5$. We find that as uncertainty is introduced and $r$ increases, the availability of LDI steps diminishes significantly (Figure \ref{fig:perc_dyns}) to the point of eliminating (almost) all cycles and shortening the path length to convergence (Figure \ref{fig:steps_to_converge}). 
Figure \ref{fig:num_cycles} presents the number of profiles whose LDI sequence cycles for no uncertainty, while only $5$ of our sampled $r \geq 1$ profiles' sequences cycle. Therefore, cycles with uncertainty are the exception rather than the norm.

These findings corroborate our theoretical analysis. As uncertainty increases, more issues are perceived by agents to have more than one possible winner. Since issues are interdependent for many preference rankings, fewer agents have LDI steps. On the other hand, as $n$ increases, more agents have rankings without these interdependencies, thus increasing the availability of LDI steps.




As an additional inquiry, we studied how IV affects the quality of outcomes by comparing the social welfare of equilibrium to truthful vote profiles.\footnote{Measured by the percent change in Borda welfare. The \emph{Borda utility} of outcome $a$ for ranking $R$ is $2^p$ minus the index of $a$'s position in $R$; the \emph{Borda welfare} is the sum of utilities across agents.} 
We find in Figure \ref{fig:welf_perc_change} that IV improves average welfare, but at a rate decreasing in $r$. This finding agrees with experiments by \citet{bowman14:potential}, \citet{grandi20:voting}, suggesting that IV may reduce multiple-election paradoxes by helping agents choose better outcomes. 
However, further work will be needed to generalize this conclusion, as it contrasts experiments of single-issue IV by \citet{meir2020strategic}, \citet{Koolyk17:Convergence}.

\section{Discussion and Open Questions}
\label{sec:conclusion}

We have introduced a novel model of strategic behavior in IV for multiple issues under uncertainty. 
We find that for binary issues, the existence of cycles hinges on the interdependence of issues in agents' preference rankings. 
Specifically, once an agent $j$ takes an LDI step on an issue $i$, they only subsequently revert their vote when their preference for $i$ changes. This occurs when the set of possible winning candidates of other issues, that affect $j$'s preference for $i$, have changed.

Without this interdependence, agents' preference over issues change a finite number of times, so LDI dynamics converge (Theorem \ref{thm:o_legal_converge}). As uncertainty increases over issues other than the one agents are changing, fewer preference rankings admit LDI steps which eliminates cycles (Theorem \ref{thm:dynamic_uncertainty_binary_converge}).
Convergence does not extend to multi-candidate issues, as LDI plurality dynamics may cycle if agents only have partial order preference information (Corollary \ref{cor:cycle_partial_order}). Our experiments confirm that convergence is practically guaranteed with uncertainty, despite its possibility, and suggests IV improves agents' social welfare over their truthful vote outcome.

There are several open directions for future work. First, our empirical study was limited by assuming agents' preferences were sampled from the impartial culture preference distribution and had additive social welfare (see e.g., \citet{Tsetlin2003:The-impartial}, \citet{Sen1999:The-Possibility}). Proving IV's welfare properties for more realistic preference distributions and welfare functions may follow the dynamic price of anarchy work of \citet{Branzei13:How}, \citet{Kavner21:strategic} and smoothed analysis techniques of \citet{Xia2020:The-Smoothed}.

Second, IV is useful, in part, for protecting agents' privacy, as they do not explicitly reveal their truthful preferences. However, agents that make improvement steps implicitly reveal partial preference information. Studying IV when agents account for others' preferences  based on current information is an interesting open direction.
A third direction is to detail the properties, such as anonymity and Condorcet-consistency, of multi-issue IV rules that may be inherited by the decision rule used locally on each issue, as studied for sequential voting rules by \citet{Xia11:Strategic}.


\bibliographystyle{plainnat}
\bibliography{for_arxiv_references}

\clearpage

\appendix

\section*{Appendix}

\section{Deferred Proofs}
\label{apx:convergence_proofs}

\setcounter{thm}{0}
\begin{thm}
Given binary issues, consider agent $j$ changing their vote on issue $i$ in vote profile $a$ with one of three uncertainty parameters: $r_j$, $q_j$, or $\hat{q}_j$. The parameters $r_j$ and $q_j$ only differ on issue $k \neq i$ such that $r_j^k < q_j^k$; $r_j$ and $\hat{q}_j$ only differ on issue $i$ such that $r_j^i < \hat{q}_j^i$. Let $LD^i_j$, $LQ^i_j$, and $\hat{LQ}^i_j$ denote $j$'s possible LDI steps according to each of the parameters. Then $LQ_j^i \subseteq LD_j^i \subseteq \hat{LQ}^i_j$.
\end{thm}

\begin{proof}
We prove the theorem by proving a weaker claim that holds for the more general multi-candidate issues setting. Specifically, let $D^i_j$, $Q^i_j$, and $\hat{Q}^i_j$ denote the set of votes that $S$-dominate $a_j$ and only differ on the $i^{th}$ issue, corresponding to $S = \tilde{S}_{-j}(a;r_j)$, $S = \tilde{S}_{-j}(a;q_j)$, or $S = \tilde{S}_{-j}(a;\hat{q}_j)$ respectively. We show that $Q^i_j \subseteq D^i_j \subseteq \hat{Q}^i_j$. The theorem follows because $LD^i_j = D^i_j$ (respectively, $LQ^i_j = Q^i_j$ and $\hat{LQ}^i_j = \hat{Q}^i_j$) for binary issues.

\underline{(Part 1: $Q^i_j \subseteq D^i_j$.)} 
Without loss of generality let $D^i_j \neq \emptyset$. Suppose $\hat{a}_j \in Q^i_j$.
Then $\hat{a}_j$ $\tilde{S}_{-j}(a;q_j)$-dominates $a_j$, entailing:
\begin{enumerate}
    \item $\exists \tilde{s} \in \tilde{S}_{-j}(a; q_j)$ such that $f(\tilde{s} + \hat{a}_j) \succ_j f(\tilde{s} + a_j)$, and
    \item $\nexists \tilde{s}' \in \tilde{S}_{-j}(a; q_j)$ such that $f(\tilde{s}' + \hat{a}_j) \prec_j f(\tilde{s}' + a_j)$.
\end{enumerate}
First, by construction of the uncertainty sets, $\tilde{S}^k_{-j}(a; r_j) \subseteq \tilde{S}^k_{-j}(a; q_j)$ and $\tilde{S}^h_{-j}(a; r_j) = \tilde{S}^h_{-j}(a; q_j)$ for all $h \neq k$; therefore $\tilde{S}_{-j}(a; r_j) \subseteq \tilde{S}_{-j}(a; q_j)$. It follows that $\nexists \tilde{s}' \in \tilde{S}_{-j}(a; r_j)$ such that $f(\tilde{s}' + \hat{a}_j) \prec_j f(\tilde{s}' + a_j)$ by point (2) above. Hence, $a_j$ does not $\tilde{S}_{-j}(a; r_j)$-beat $\hat{a}_j$.

Second, define a score tuple $\tilde{v}$ such that $\tilde{v}^h = \tilde{s}^h$ for each $h \neq k$ and $\tilde{v}^k \in \tilde{S}^k(a; r_j)$ arbitrarily. It is the case that $\tilde{v} \in \tilde{S}(a; r_j)$ since $\tilde{S}^h_{-j}(a; r_j) = \tilde{S}^h_{-j}(a; q_j)$ for all $h \neq k$. Since $f^i(\tilde{s} + \hat{a}_j) \neq f^i(\tilde{s} + a_j)$ and $f^h(\tilde{s} + \hat{a}_j) = f^h(\tilde{s} + a_j)$ for all $h \neq i$, we have $f(\tilde{v} + \hat{a}_j) \neq f(\tilde{v} + a_j)$. It follows from the above argument that $f(\tilde{v} + \hat{a}_j) \succ_j f(\tilde{v} + a_j)$. Therefore $a_j~\tilde{S}_{-j}(a; r_j)$-beats $\hat{a}_j$ and $\hat{a}_j \in D^i_j$.

\underline{(Part 2: $D^i_j \subseteq \hat{Q}^i_j$.)} 
Without loss of generality let $\hat{Q}^i_j \neq \emptyset$. Suppose $\hat{a}_j \in D^i_j$.
Then $\hat{a}_j$ $\tilde{S}_{-j}(a;r_j)$-dominates $a_j$, entailing:
\begin{enumerate}
    \item $\exists \tilde{s} \in \tilde{S}_{-j}(a; r_j)$ such that $f(\tilde{s} + \hat{a}_j) \succ_j f(\tilde{s} + a_j)$, and
    \item $\nexists \tilde{s}' \in \tilde{S}_{-j}(a; r_j)$ such that $f(\tilde{s}' + \hat{a}_j) \prec_j f(\tilde{s}' + a_j)$.
\end{enumerate}

By construction of the uncertainty sets, $\tilde{S}^i_{-j}(a; r_j) \subseteq \tilde{S}^i_{-j}(a; \hat{q}_j)$ and $\tilde{S}^h_{-j}(a; r_j) = \tilde{S}^h_{-j}(a; \hat{q}_j)$ for all $h \neq i$; therefore $\tilde{S}_{-j}(a; r_j) \subseteq \tilde{S}_{-j}(a; \hat{q}_j)$. 
It immediately follows that $\tilde{s} \in \tilde{S}_{-j}(a; \hat{q}_j)$, so $\hat{a}_j$ $\tilde{S}_{-j}(a; r_j)$-beats $a_j$.
Furthermore, $\nexists \tilde{s}' \in \tilde{S}_{-j}(a; \hat{q}_j)$ such that $f(\tilde{s}' + \hat{a}_j) \prec_j f(\tilde{s}' + a_j)$ by point (2) above. Hence, $a_j$ does not $\tilde{S}_{-j}(a; r_j)$-beat $\hat{a}_j$.
Therefore $\hat{a}_j \in \hat{Q}^i_j$, concluding our proof.

\end{proof}

\setcounter{lem}{0}

\begin{lem}
Consider an LDI step $a_j \xrightarrow{j} \hat{a}_j$ over issue $i$ from vote profile $a$ by agent $j$ with uncertainty parameter $r_j$. Then either (1) $a^i_j \notin H^i_j(a)$; or for every combination of possible winners in $\calW^{-i}(a; r_j)$, either (2) $a^i_j \prec_j b$ for all $b \in H^i_j(a)$; or (3) $r^i_j = 0$, $\{a^i_j, \hat{a}^i_j\} \subseteq H_0(a)$ and $\hat{a}^i_j \succ_j a^i_j$.
\end{lem}


\begin{proof}
The proof follows that of Lemma 3 in \citet{Meir15:Plurality}. Suppose that $a^i_j, b \in H^i_j(a)$ and $a^i_j \succ_j b$ whenever some combination of possible winners $(c^1, c^2, \ldots, c^{i-1}, c^{i+1}, \ldots, c^p) \in \calW^{-i}(a; r_j)$ wins (i.e., (1) and (2) are violated). Assume first that $\hat{a}^i_j \notin H_0(a)$. By Lemma 2 of \citet{Meir15:Plurality}, $\exists \tilde{s} \in \tilde{S}_{-j}(a; r_j)$ such that:
\begin{itemize}
    \item $a^i_j, b$ have maximal score (possibly with other candidates), strictly above $\hat{a}^i_j$, in $\tilde{s}^i$;
    \item for each $k \neq i$, $\tilde{s}^k$ is such that $f^k(\tilde{s}+a_j) = c^k$ wins.
\end{itemize} 
W.l.o.g. assume $b$ is prior to $a^i_j$ in tie-breaking (otherwise adjust $\tilde{s}$ so that $f^i(\tilde{s})=b$).
Thus $f^i(\tilde{s}+a_j) = a^i_j$ while $f^i(\tilde{s}+\hat{a}_j) = b$. Since $a^i_j \succ_j b$ given $\tilde{s}$, this implies $a_j~\tilde{S}_{-j}(a; r_j)$-beats $\hat{a}_j$.

The remaining case is where $\hat{a}^i_j \in H_0(a)$ and $a^i_j \succ_j \hat{a}^i_j$ whenever some $(c^1, c^2, \ldots, c^{i-1}, c^{i+1}, \ldots, c^p) \in \calW^{-i}(a; r_j)$ wins. Then in $\tilde{s}$ where $a^i_j, \hat{a}^i_j$ are tied and $(c^k)_{k \in \calP \backslash \{i\}}$ wins, it is better for $j$ to vote for $a_i$.

In both cases we get $\hat{a}_j \notin LD^i_j$, which is a contradiction.
\end{proof}


\section{Supplementary Examples}
\label{apx:ex}

\setcounter{ex}{3}
\begin{ex}{(Cont'd.)}
Recall the agents $j$ and $k$ with the following preferences:
\begin{itemize}
    \item If $f^1 = 0$ then $j$ prefers $b \succ c \succ a \succ d$ on the second issue; 
    \item if $f^1 = 1$ then $j$ prefers $c \succ b \succ a \succ d$ on the second issue;
    \item agent $k$ prefers $a \succ d \succ b \succ c$ on the second issue regardless of the first issue's outcome.
\end{itemize}

Initialize the vote profile $a(0)$ such that $s(a(0)) = \{(7,8), (3,5,5,2)\}$ and $a_j(0) = a_k(0) = (0,a)$. There are four LDI steps involved in this cycle: (i) $(0,a) \xrightarrow{j} (0,d)$, (ii) $(0,a) \xrightarrow{k} (0,d)$, (iii) $(0,d) \xrightarrow{j} (0,a)$, and (iv) $(0,d) \xrightarrow{k} (0,a)$.
We will prove that these are valid LDI steps in turn, demonstrating that (i) the new vote $S$-beats the old vote, (ii) the old vote does not $S$-beat the new vote, and (iii) the new vote is not $S$-dominated. For any candidate $e \in \{a,b,c,d\}$, denote the vote switching the second candidate to $e$ by $\hat{e} = (0,e)$. Recall that ties are broken in lexicographical order.

\underline{(Step 1: $s^2(a(0)) = (3,5,5,2)$.)} Since $r^2_j=1$ and $a^2_j = a$, we have $H^2_j(a(0)) = \{a,b,c\}$.

(i) Let $\tilde{s} = \{(6,8),(3,4,4,2)\} \in \tilde{S}_{-j}(a(0); r_j)$. Then $f(\tilde{s} + \hat{d}) = (1,b) \succ_j (1,a) = f(\tilde{s} + \hat{a})$, so $\hat{d} ~\tilde{S}_{-j}$-beats $\hat{a}$.

(ii) For any $\tilde{s} \in \tilde{S}_{-j}(a(0); r_j)$, we have that if $f^2(\tilde{s}+\hat{d}) = a$ then $f^2(\tilde{s}+\hat{a}) = a$; otherwise, if $f^2(\tilde{s}+\hat{a}) \in \{b,c\}$ then $f^2(\tilde{s}+\hat{d}) = f^2(\tilde{s}+\hat{a})$. Hence, it is never preferable for $j$ to vote for $\hat{a}$ than $\hat{d}$.

(iii) Neither $\hat{b}$ nor $\hat{c}$ $\tilde{S}_{-j}$-dominate $\hat{d}$ since for $\tilde{s}' = \{(6,8), (2,4,5,2)\}$ and $\tilde{s}'' = \{(7,7), (2,4,4,2)\}$ it is preferable for $j$ to vote for $\hat{d}$ than $\hat{b}$ and $\hat{c}$, respectively.

\underline{(Step 2: $s^2(a(1)) = (2,5,5,3)$.)} Since $r^2_k=1$ and $a^2_j = a$, we have $H^2_k(a(1)) = \{b,c,d\}$.

(i) Let $\tilde{s} = \{(6,8),(2,4,4,4)\} \in \tilde{S}_{-k}(a(1); r_k)$. Then $f(\tilde{s} + \hat{d}) = (1,d) \succ_k (1,b) = f(\tilde{s} + \hat{a})$, so $\hat{d} ~\tilde{S}_{-k}$-beats $\hat{a}$.

(ii) For any $\tilde{s} \in \tilde{S}_{-k}(a(1); r_k)$, we have that if $f^2(\tilde{s}+\hat{d}) = d$ then $f^2(\tilde{s}+\hat{a}) = b$; otherwise $f^2(\tilde{s}+\hat{d}) = f^2(\tilde{s}+\hat{a})$. Hence, it is never preferable for $k$ to vote for $\hat{a}$ than $\hat{d}$.

(iii) Neither $\hat{b}$ nor $\hat{c}$ $\tilde{S}_{-k}$-dominate $\hat{d}$ since for $\tilde{s} = \{(6,8), (2,4,4,4)\}$ (the same as in (i)), it is preferable for $j$ to vote for $\hat{d}$ than either $\hat{b}$ or $\hat{c}$.

\underline{(Step 3: $s^2(a(2)) = (1,5,5,4)$.)} Since $r^2_j=1$ and $a^2_j = d$, we have $H^2_j(a(2)) = \{b,c,d\}$.

(i) Let $\tilde{s} = \{(6,8),(1,4,4,4)\} \in \tilde{S}_{-j}(a(2); r_j)$. Then $f(\tilde{s} + \hat{a}) = (1,b) \succ_j (1,d) = f(\tilde{s} + \hat{d})$, so $\hat{a} ~\tilde{S}_{-j}$-beats $\hat{d}$.

(ii) For any $\tilde{s} \in \tilde{S}_{-j}(a(2); r_j)$, we have that if $f^2(\tilde{s}+\hat{a}) = d$ then $f^2(\tilde{s}+\hat{d}) = d$; otherwise, if $f^2(\tilde{s}+\hat{d}) \in \{b,c\}$ then $f^2(\tilde{s}+\hat{a}) = f^2(\tilde{s}+\hat{d})$. Hence, it is never preferable for $j$ to vote for $\hat{d}$ than $\hat{a}$.

(iii) Neither $\hat{b}$ nor $\hat{c}$ $\tilde{S}_{-j}$-dominate $\hat{a}$ since for $\tilde{s}' = \{(6,8), (2,4,5,3)\}$ and $\tilde{s}'' = \{(7,7), (2,4,4,3)\}$ it is preferable for $j$ to vote for $\hat{a}$ than $\hat{b}$ and $\hat{c}$, respectively.

\underline{(Step 4: $s^2(a(3)) = (2,5,5,3)$.)} Since $r^2_k=1$ and $a^2_j = d$, we have $H^2_k(a(3)) = \{a,b,c\}$.

(i) Let $\tilde{s} = \{(6,8),(3,4,4,2)\} \in \tilde{S}_{-k}(a(3); r_k)$. Then $f(\tilde{s} + \hat{a}) = (1,a) \succ_k (1,b) = f(\tilde{s} + \hat{d})$, so $\hat{a} ~\tilde{S}_{-k}$-beats $\hat{d}$.

(ii) For any $\tilde{s} \in \tilde{S}_{-k}(a(3); r_k)$, we have that if $f^2(\tilde{s}+\hat{a}) = a$ then $f^2(\tilde{s}+\hat{d}) = b$; otherwise $f^2(\tilde{s}+\hat{a}) = f^2(\tilde{s}+\hat{d})$. Hence, it is never preferable for $k$ to vote for $\hat{d}$ than $\hat{a}$.

(iii) Neither $\hat{b}$ nor $\hat{c}$ $\tilde{S}_{-k}$-dominate $\hat{a}$ since for $\tilde{s} = \{(6,8), (3,4,4,2)\}$ (the same as in (i)), it is preferable for $j$ to vote for $\hat{a}$ than either $\hat{b}$ or $\hat{c}$.
\end{ex}


\begin{ex}
Consider $p=2$ issues with three candidates each, labeled $\{A^i, B^i, C^i, D^i\}$ for $i \leq 2$. Let agent $j$ have $\ell_\infty$ uncertainty metric and $\mo=(2,1)$-legal preferences that satisfy:
\begin{itemize}
    \item If $f^2 = C^2$, then over issue $1$: $A^1 \succ_j B^1 \succ_j D^1 \succ_j C^1$
    \item If $f^2 = B^2$, then over issue $1$: $B^1 \succ_j A^1 \succ_j D^1 \succ_j C^1$
    \item If $f^2 = A^2$, then over issue $1$: $C^1 \succ_j A^1 \succ_j B^1 \succ_j D^1$
\end{itemize}
Suppose $a_j = (D^1, C^2)$ in vote profile $a$ with score tuple $s(a) = \{(6,6,3,5),(2,4,6,0)\}$.

\underline{(Part 1: increasing $r^1_j$.)} This part of the example is motivated by Lemma 3b of \citet{Meir15:Plurality}.

If $r_j = (0,0)$ then it is easy to see that $j$ has no BR steps: the unique winner is $f(a) = (A^1, C^2)$ and, although $j$ can make $B^1$ win, they prefer $A^1 \succ B^1$ when $f^2(a) = C^2$. Hence $LD^1_j= \emptyset$. 

If $r_j = (1,0)$ then
\[
\tilde{S}^1_{-j} = \{(5,6,7) \times (5,6,7) \times (2,3,4) \times (3,4,5)\}
\]
First, when $\tilde{s} = \{(5,5,4,5), (2,4,6,0)\} \in \tilde{S}_{-j}$, it is preferable for $j$ to vote for either $\hat{a}_j = (A^1, C^2)$ or $\hat{a}'_j = (B^1, C^2)$ instead of $a_j$. Recall that ties are broken by lexicographical order. Hence both $\hat{a}_j$ and $\hat{a}'_j ~\tilde{S}_{-j}$-beat $a_j$. 
Second, it is easy to see that there is no $\tilde{s}' \in \tilde{S}_{-j}$ where it is preferable for $j$ to vote for $a_j$ instead of either $\hat{a}_j$ or $\hat{a}'_j$. Thus, $a_j$ does not $\tilde{S}_{-j}$-beat either $\hat{a}_j$ or $\hat{a}'_j$.
Third, we notice that $\hat{a}_j~\tilde{S}$-dominates $\hat{a}'_j$, which which follows from $A^1 \succ_j B^1$ whenever $f^2(a)=C^2$. It is easy to see that $\hat{a}_j$ is not $\tilde{S}$-dominated. Thus $LD^1_j = \{(A^1,C^2)\}$.

Now consider $r_j = (2,0)$. Then
\begin{align*}
\tilde{S}^1_{-j} = \{&(4,5,6,7,8) \times (4,5,6,7,8) \\
\times & (1,2,3,4,5) \times (2,3,4,5,6)\}
\end{align*}
Since $\tilde{s} = \{(5,5,3,5), (2,4,6,0)\} \in \tilde{S}_{-j}$, both $\hat{a}_j$ and $\hat{a}'_j$ still $\tilde{S}_{-j}$-beat $a_j$. However, as $\tilde{s} = \{(4,4,5,5), (2,4,6,0)\} \in \tilde{S}_{-j}$ we have $f^1(\tilde{s}+a_j) = (D^1,C^2) \succ_j (C^1,C^2) = f^1(\tilde{s}+\hat{a}_j) = f^1(\tilde{s}+\hat{a}'_j)$, so $a_j ~\tilde{S}_{-j}$-beats both $\hat{a}_j$ and $\hat{a}'_j$. Hence, $LD^1_j = \emptyset$.

\underline{(Part 2: increasing $r^2_j$.)} Recall that when $r_j = (1,0)$, we found that $\hat{a}_j = (A^1,C^2) ~\tilde{S}_{-j}$-dominated $a_j = (D^1,C^2)$, while $\hat{a}'_j = (B^1,C^2) ~\tilde{S}_{-j}$-dominated $a_j$ and $\hat{a}_j ~\tilde{S}_{-j}$-dominated $\hat{a}'_j$; thus $LD^1_j = \{(A^1, C^2)\}$.

If $r_j = (1,1)$, then
\[
\tilde{S}^2_{-j} = \{(1,2,3) \times (3,4,5) \times (4,5,6) \times (0,1)\}
\]
For $\tilde{s} = \{(5,5,4,5),(2,4,5,0)\} \in \tilde{S}_{-j}$ it is preferable for $j$ to vote for $\hat{a}_j$ rather than $\hat{a}'_j$, whereas for $\tilde{s} = \{(5,5,4,5),(2,5,4,0)\} \in \tilde{S}_{-j}$ it is preferable for $j$ to vote for $\hat{a}'_j$ rather than $\hat{a}_j$. Then both $\hat{a}_j$ and $\hat{a}'_j$ $\tilde{S}_{-j}$-dominate $a_j$ but neither are $\tilde{S}_{-j}$-dominated. Therefore $LD^1_j = \{(A^1, C^2), (B^1, C^2)\}$.

If $r_j = (1,2)$, then
\begin{align*}
\tilde{S}^2_{-j} = \{&(0,1,2,3,4) \times (2,3,4,5,6) \\
\times & (3,4,5,6,7) \times (0,1,2)\}
\end{align*}
Then there are score vectors $\tilde{s} \in \tilde{S}_{-j}$ where $f^2(\tilde{s}+a_j) = f^2(\tilde{s}+\hat{a}_j) = f^2(\tilde{s}+\hat{a}'_j) = C^2$, so that neither $\hat{a}_j$ nor $\hat{a}'_j~\tilde{S}_{-j}$-dominate $a_j$. Thus $LD^1_j = \emptyset$.

\label{ex:LDI_steps_different_uncertainties}
\end{ex}

\begin{ex}
Let there be $p=2$ binary issues and $n=13$ agents who each use the $\ell_\infty$ uncertainty metric with common fixed uncertainty parameters $(r^1_j, r^2_j) = (1,2)~\forall j \leq n$. Suppose that agents' preferences abide by the following four types:
\begin{itemize}
    \item (Type 1) three agents have rankings: 
$(0,1) \succ (1,1) \succ (1,0) \succ (0,0)$; 
    \item (Type 2) five agents have rankings: $(0,0) \succ (0,1) \succ (1,1) \succ (1,0)$; 
    \item (Type 3) four agents have rankings: $(1,0) \succ (1,1) \succ (0,0) \succ (0,1)$; and
    \item (Type 4) one agent has ranking: $(1,1) \succ (1,0) \succ (0,1) \succ (0,0)$.
\end{itemize}

There is a cycle passing through the four vote profiles $a(0)$ (which is truthful), $a(3)$, $a(8)$, and $a(11)$ listed in Table \ref{tab:LDI_cycle}, in which every agent of the same type has the same vote. There are four parts of the cycle between these profiles:
\begin{itemize}
    \item between $a(0)$ and $a(3)$, all agents of Type 1 make LDI steps on the first issue $(0,1) \xrightarrow{1} (1,1)$;
    \item between $a(3)$ and $a(8)$, all agents of Type 2 make LDI steps on the second issue $(0,0) \xrightarrow{2} (0,1)$;
    \item between $a(8)$ and $a(11)$, all agents of Type 1 make LDI steps on the first issue $(1,1) \xrightarrow{1} (0,1)$;
    \item between $a(11)$ and $a(16)=a(0)$, all agents of Type 2 make LDI steps on the second issue $(0,1) \xrightarrow{2} (0,0)$.
\end{itemize}
Notice that no agent of Types 3 and 4 make LDI steps since they are voting truthfully and have separable preferences.

\begin{table}[h!]
\centering
 \begin{tabular}{|c|c|c|c|c|} 
 \hline
 Agent Type $j$ & $a_j(0)$ & $a_j(3)$ & $a_j(8)$ & $a_j(11)$ \\ [0.5ex]
 \hline
 $1$ & $(0,1)$ & $(1,1)$ & $(1,1)$ & $(0,1)$ \\
 $2$ & $(0,0)$ & $(0,0)$ & $(0,1)$ & $(0,1)$ \\
 $3$ & $(1,0)$ & $(1,0)$ & $(1,0)$ & $(1,0)$ \\
 $4$ & $(1,1)$ & $(1,1)$ & $(1,1)$ & $(1,1)$ \\
 \hline
 \end{tabular}
 \caption{Agents' votes for $a(0)$ (truthful), $a(3)$, $a(8)$, and $a(11)$.}
 \label{tab:LDI_cycle}
\end{table}

We claim these are valid LDI steps as follows. For any vote profile in $\{a(0), a(1), a(2)\}$, let $t$ be the number of Type 1 agents who have made LDI steps. Then $s(a(t)) = \{(8-t,5+t), (9,4)\}$ and for any Type 1 agent $j$ who has not made their first LDI step yet,

$
\tilde{S}^1_{-j} = \{(6-t, 7-t, 8-t) \times (4+t, 5+t, 6+t)\}
$

$
\tilde{S}^2_{-j} = \{(7,8,9,10,11) \times (2,3,4,5,6)\}
$

Define $\hat{a}_j = (1,1)$. Then $\tilde{s} = \{(6,6), (9,4) \in \tilde{S}_{-j}(a(t); r_j)~\forall t \in \{0,1,2\}$, so $f(\tilde{s}+ \hat{a}_j)=(1,0) \succ_j (0,0) = f(\tilde{s}+a_j)$ and $\hat{a}_j~\tilde{S}_{-j}$-beats $a_j$. It is easy to see that $a_j$ does not $\tilde{S}_{-j}$-beat $\hat{a}$ and that $\hat{a}$ is not $\tilde{S}_{-j}$-dominated since issues are binary. Thus $LD^1_j = \{\hat{a}_j\}$.

The other LDI steps follow similar reasoning, yielding the cycle presented in the table. It can be verified that this represents all possible LDI sequences from the truthful vote profile.

\label{ex:LDI_cycle}
\end{ex}


\section{Nonatomic model}
\label{apx:non_atomic_model}

Each of the results presented for binary issues and atomic agents, where each agent contributes one unit of influence to the population of $n$ votes, also generalize to a nonatomic variant of IV, in which agents are part of a very large population and have negligible influence over the outcome. Our model extends \citet{Meir15:Plurality}'s iterative plurality voting for nonatomic agents to the multi-issue setting. Like this setting, our convergence results permit arbitrary subsets of agents to change their vote simultaneously. This differs from the finite case which may not converge if multiple agents change their votes simultaneously \cite{Meir10:Convergence}. 

In this section, we provide the necessary definitions for the nonatomic model, using our existing notation wherever possible.
There are two major differences: (i) with the identify of an agent, and (ii) with the uncertainty score set. First, rather than treating each agent $j \leq n$ individually, there are groups of identical agents with $\epsilon$ mass; ``an agent of type $j$'' then refers to a representative agent in this group. Second, as agents have negligible influence, type $j$ agents consider any score tuple in $\tilde{S}(a; r_j)$ possible; all agents agree of the same type agree what real score tuples are possible.
Lemma \ref{lem:lem3_multi} and Theorems \ref{thm:improvement_step_relations}, \ref{thm:o_legal_converge}, and \ref{thm:dynamic_uncertainty_binary_converge} are each upheld after applying this model redefinition. Similar non-convergence results also apply for the nonatomic variant of IV.


\begin{paragraph}{Basic notations.}
We do not have a finite set of agents. Rather, a \emph{preference profile} $Q \in \Delta(\calL(D))$ is a distribution over rankings that specifies the fraction of agents $Q(R)$ for each $R\in \calL(D)$.
We only consider moves by subsets of agents since each agent has negligible influence. To avoid infinite improvement sequence paths of sequentially smaller subsets of agents, we assume there is a minimum resolution $\epsilon$, such that sets of agents of the same type with mass $\epsilon$ always move together (although in an uncoordinated manner; see Appendix~IX of \citet{Meir15:Plurality}).
We denote the collection of these $1/\epsilon$ sets by $J$. Since all agents in set $j \in J$ are indistinguishable, we refer to ``agent $j$'' as an arbitrary agent in the set $j$. Then $R_j\in \calL(D)$ is the preference, $r_j$ is the uncertainty parameter, and $a_j$ is the vote of an arbitrary agent in the set of vote profile $a$.
\end{paragraph}

\begin{paragraph}{Winner determination.} 
For any issue $i \in \calP$, we define the \emph{score vector} $s^i(a) = (s^i(c;a))_{c \in D_i}$ induced by the vote profile $a$ such that $s^i(c;a) = |\{f:a^i_j = c\}| \epsilon \in [0,1]$.
Winner determination is exactly as in the atomic model with lexicographical tie-breaking.  
\end{paragraph}

\begin{paragraph}{Uncertainty and local dominance.} 
We assume agents utilize the same candidate-wise distance uncertainty as the atomic model. Like the atomic model, each agent of type $j \in J$ selected by the scheduler to change their vote has uncertainty parameters $r_j = (r^i_j)_{i \in \calP}$. Unlike the atomic model, agents have negligible influence in the score vector. Hence, they take their uncertainty score sets with respect to the real score tuple:
$\tilde{S}^i(a; r^i_j) = \{v^i : \delta(v^i, s^i(a)) \leq r^i_j \}$. As before, $\tilde{S}(a; r_j) = \times_{i \in \calP} \tilde{S}^i(a; r^i_j)$.

Given $v \in \tilde{S}(a; r_j)$, we define $f(v+a_j)$ to be the outcome an agent of type $j$ expects by voting for $a_j$ according to score vector $s$. That is, for each issue $i$, the extra vote $a^i_j$ decides the winner if several candidates are tied with maximal score in $s$, overriding the default tie-breaker (see Appendix X of \citet{Meir15:Plurality}).

The definition of a local dominance improvement (LDI) step for a nonatomic agent of type $j \in J$ from vote $a_j$ to $\hat{a}_j$ on issue $i$ then is the same as in the atomic model of Definitions \ref{dfn:s_dom} and \ref{dfn:local_dom}, applying this redefinition and using $\tilde{S}(a; r_j)$ in lieu of $\tilde{S}_{-j}(a; r_j)$.
The response function $g_j$ is also the same, except that its domain (all possible profiles) is now $\calD^{|J|}$ rather than $\calD^{n}$. The definition of equilibrium does not change.

\end{paragraph}

\end{document}